\documentclass[5p]{elsarticle}

\usepackage{amsmath}
\usepackage{amsfonts}
\usepackage{sistyle}
\usepackage{float}
\usepackage{graphicx}
\usepackage{diagbox}
\usepackage{setspace}
\usepackage[switch]{lineno}
\usepackage{lipsum}
\usepackage[utf8]{inputenc}
\usepackage[english,onelanguage,ruled,vlined]{algorithm2e}
\usepackage[noend]{algorithmic}
\usepackage[dvipsnames]{xcolor}

\DeclareMathOperator*{\argmin}{arg\,min}

\newtheorem{observation}{Observation}

\newenvironment{proof}{\hspace{10pt}\textit{Proof}:}{\hfill$\square$}

\begin{document}
	
	

\title{Critical Time Windows for Renewable Resource Complementarity Assessment}

\author[ulgm]{Mathias~Berger\corref{cor1}\fnref{fn1}}
\ead{mathias.berger@uliege.be}
\author[ulgm]{David~Radu\fnref{fn1}}
\author[ulgm]{Raphaël~Fonteneau}
\author[ulgm]{Robin~Henry}
\author[ulgm]{Mevludin~Glavic}
\author[ulgg]{Xavier~Fettweis}
\author[rte]{Marc~Le~Du}
\author[rte]{Patrick~Panciatici}
\author[rte]{Lucian~Balea}
\author[ulgm]{Damien~Ernst}

\cortext[cor1]{Corresponding author.}
\fntext[fn1]{Equally contributing authors.}
\address[ulgm]{Department of Electrical Engineering and Computer Science, University of Liège, Allée de la Découverte 10, 4000 Liège, Belgium}
\address[ulgg]{Laboratory of Climatology, Department of Geography, University of Liège, Belgium}
\address[rte]{R\&D Department, Réseau de Transport d'Electricité (RTE), France}

\begin{abstract}
This paper proposes a systematic framework to assess the complementarity of renewable resources over arbitrary geographical scopes and temporal scales which is particularly well-suited to exploit very large data sets of climatological data. The concept of critical time windows is introduced, and a spatio-temporal criticality indicator is proposed, consisting in a parametrised family of scalar indicators quantifying the complementarity between renewable resources in both space and time. The criticality indicator is leveraged to devise a family of optimisation problems identifying sets of locations with maximum complementarity under arbitrary geographical deployment constraints. The applicability of the framework is shown in a case study investigating the complementarity between the wind regimes in continental western Europe and southern Greenland, and its usefulness in a power system planning context is demonstrated. Besides showing that the occurrence of low wind power production events can be significantly reduced on a regional scale by exploiting diversity in local wind patterns, results highlight the fact that aggregating wind power production sites located on different continents may result in a lower occurrence of system-wide low wind power production events and indicate potential benefits of intercontinental electrical interconnections.

\end{abstract}


\maketitle
	

\section{Introduction}
\label{sec:introduction}
The recent large-scale deployment of variable renewable energy (VRE) technologies to produce electricity has resulted in novel operational challenges in the power system, as their locally intermittent nature make it difficult to reliably supply inflexible loads on time scales ranging from seconds to hours and days. Solution avenues addressing these issues range from increasing storage capacity in the system to developing transmission capacity and interconnections on a continental or global scale to exploit the apparent smoothing effect occurring in renewable production patterns aggregated over vast geographical areas and hence reduce the variability of the aggregate production \cite{GiebelPhD, LouieChapter}, as well as to take advantage of time differences in production and consumption profiles in different regions of the world. The latter approach has received increasing attention in recent years \cite{chatzivasileiadis2013global, chatzivasileiadis2017global}, but few practical quantitative methods to assess its real benefits have been proposed.

This paper proposes a methodology to assess the complementarity of renewable resources over arbitrary geographical scopes and time scales, with an eye towards the planning of large international interconnections. The methodology considers a set of locations for which simultaneously and uniformly sampled time series describing stochastic signals of interest are available, from which a set of time windows of arbitrary length is extracted. Various resource quality criteria are then invoked to identify time windows characteristic of specific events deemed critical for each location considered, such as an averaged capacity factor value smaller than a given threshold over the window duration. The concept of critical time windows is then leveraged to derive a criticality indicator quantifying the complementarity between locations in terms of the proportion of critical time windows common to a pre-specified proportion of these locations. This framework is further utilised to derive optimisation problems identifying the set of locations having the best complementarity within a given geographical area, possibly considering constraints on the number of sites that can be selected in sub-regions of interest. Compared to traditional statistical or frequency domain approaches evaluating the complementarity between renewable resources and the properties of their aggregate, this methodology exploits a precise description of signals in the time-domain and keeps track of their chronology, making it ideally-suited to accurately capture extreme system-wide events having a substantial impact on system design, especially on storage and transmission capacity sizing, such as long periods over which very little renewable generation is available, which cannot usually be handled by storage alone. It also easily lends itself to optimal generation site selection, and thus constitutes a powerful and versatile tool which can be employed in renewable resource complementarity analyses, as well as in planning studies to identify the optimal deployment patterns of renewable generation sites and associated interconnection needs.

The methodology is illustrated on a case study evaluating the potential of and assessing the complementarity between the wind regimes in France and southern Greenland, in the hypothetical context of the development of an interconnection between continental western Europe and Greenland, as part of a broader global grid initiative. Wind speed time series are transformed into power output time series using a transfer function representing the power output of a wind farm. The framework is applied to the resulting power series in order to assess whether the union of the South Greenland and France regions yields better scores of the criticality indicator than either region taken independently. Optimisation problems are also formulated to identify locations with the best complementarity. Finally, the framework is exploited to contrast the distribution of sites in France with highest capacity factors, i.e. those usually selected first for development, with that of the sites with highest complementarity.

The remainder of the article is structured as follows. Section \ref{sec:related_work} reviews related works. Thereafter, Section \ref{sec:formalisation} details the mathematical formalisation of the criticality indicator before deriving several methodologies for VRE resource quality assessment and optimal generation sites selection. Section \ref{sec:applications} presents the case study used to illustrate the usefulness of the criticality indicator, while corresponding results are presented in Section \ref{sec:results}. Lastly, Section \ref{sec:conclusion} concludes the paper and presents research questions to be addressed in future work.

\section{Related Works}
\label{sec:related_work}
The complementarity between dispersed renewable production sites, the associated reduction in renewable production variability that may be achieved by exploiting geographic diversity, and how they relate to energy security in power systems have been studied by several authors. Several different classes of methods have been employed in the literature, as reviewed next.

Standard statistical tools have most often been invoked to assess the complementarity of renewable resources. For instance, the complementarity between wind and solar resources at the regional and country scales has been investigated using descriptive statistics (percentile ranking) \cite{hoicka2011solar}, Monte Carlo methods \cite{MonfortiComplementarityAssessmentMonteCarlo}, and correlation analyses in \cite{WidenCorrelationWindSolarSweden}. In particular, correlation analyses including the use of Pearson, Spearman, and sample coefficients as well as cross-correlations have proved very popular in the literature considering the complementarity of renewable resources in various settings \cite{BettClimatologicalRelationshipsWindSolar}, \cite{JurasComplementarityReliabilitySmallHybrid}, \cite{CantaoHydroComplementarityBrazil}, \cite{DosAnjosLongTermCrossCorrelationBrazil}. An early application of these tools can be found in \cite{GiebelPhD}, which considers the case of large-scale wind farms deployment across Europe. The correlation between wind signals recorded at two different locations is computed and expressed as a function of the distance separating the wind farms. The statistical properties of spatially averaged wind signals are also investigated, uncovering the existence of a smoothing effect, whereby increasing the geographical scope appears to decrease the variability in the spatially averaged signal. More comprehensive analyses of the smoothing effect which may occur when coupling regional systems with high wind penetration through continental interconnections are presented in \cite{LouieChapter} and \cite{LouieCorrelation}. In these two papers, the correlation between regional systems on several time scales and the statistical properties of the aggregate system are studied, showing reduced variability for the aggregate system but highlighting the presence of statistically significant correlation between regional systems due to seasonal and diurnal forcing. In \cite{StMartin}, such seasonal and diurnal trends are removed from wind signals in regions of interest in order to focus on their stochastic fluctuations and study the dependence of spatial correlation scale upon fluctuation frequency. More precisely, in each region, a distance beyond which correlation assumes a small value (compared to unity), known as the correlation length, is computed for signals from which trends slower than given time constants have been removed. Overall, it is found that the correlation length increases as the frequency of fluctuations decreases, with an apparent saturation effect for phenomena occurring on time scales beyond a few hundreds of hours. In other words, high frequency fluctuations (on the order of a few hours at most) can be effectively attenuated by aggregation over small geographical scopes (a few dozens or hundreds of kilometres), whereas phenomena occurring on longer time scales require aggregation over much larger areas (hundreds or thousands of kilometres).

In \cite{KatzensteinWindVariability}, the smoothing effect that may result from integrating wind power plants over large geographical areas is considered using frequency-domain analysis methods introduced in \cite{AptPowerSpectrum} instead of statistical tools. More precisely, the level of smoothing is mapped to a feature of the power spectral density of the aggregate system, highlighting the reduction in short-term variability as the geographical scope increases. A similar study is performed for utility-scale solar power plants in Gujarat in \cite{KlimaSolarVariability}, particularly highlighting diminishing marginal returns in variability reduction as more plants are aggregated within the state. Finally, the use of scalar indices to study the complementarity between renewable resources has been proposed in \cite{BelucoComplementarityIndexSolarHydro} and \cite{LiComplementarityWindSolarOklahoma}. In particular, in \cite{LiComplementarityWindSolarOklahoma}, the index is defined as the product of three partial indices describing temporal aspects, energy content and the amplitude of the signals considered, respectively. This index is then used to quantify complementarity between wind and solar resources across Oklahoma and identify physical factors such as moisture and temperature directly influencing complementarity.

Overall, these studies bring valuable insight into the physics and benefits of aggregating intermittent renewable production plants over geographical areas of varying scopes. Two main conclusions emerge from reviewing the aforementioned literature. Firstly, it appears that aggregation at the continental level may not suffice to reduce correlation between wind plants and therefore variability beyond time scales much longer than a day. This observation highlights the need to start considering the aggregation of renewable resources scattered across several continents, and thus the interest of a power system connecting them. Secondly, none of the methods mentioned above genuinely lends itself to systematic application in a power systems planning context. Indeed, even when given sufficient data, it appears that none of them can accurately capture extreme events heavily impacting system design, let alone be easily leveraged to systematically identify optimal locations sets having the best complementarity. The methodology proposed in this paper addresses these drawbacks and can serve as a tool supporting the planning of global interconnections.

\section{Methodological Framework}
\label{sec:formalisation}
This section details the mathematical formalism underpinning the proposed methodological framework.

\subsection{Formalisation of the Criticality Indicator}
\label{sub:math_formalization}

Starting from the notion of a time window, resource quality assessment criteria are introduced, enabling the definition of so-called critical time windows. This concept is then leveraged to propose a parametrised family of scalar indicators quantifying the complementarity between different resources in space and time.

\subsubsection*{Set of locations}
\label{subsub:set_of_locations}

Let $\mathcal L$ denote the set of all locations being investigated. Each location $l \in \mathcal L $ is defined by its geographical coordinates, i.e. $ l = (\lambda^{lon} , \lambda^{lat} ) \in \mathbb R^2$, corresponding to a pair of longitude and latitude values. We also use $L$ to denote any subset $L \subset \mathcal L$, i.e. $L \in \mathcal P(\mathcal L)$, where $\mathcal P(\mathcal L) $ denotes the power set of $\mathcal L$, the set of all subsets of $\mathcal L$. In the sequel, any location $l$ is considered a potential VRE generation site.

\subsubsection*{Time windows}
\label{subsub:time_windows}

Let $ \mathcal{T} = \{ 0, \ldots, T-1  \} $ represent the discretisation of the time horizon under consideration, where $|\mathcal{T}| = T \in \mathbb N$ is the length of the latter, and let $\delta  \in \{ 1, \ldots ,T \}$ be a time duration less than or equal to $T$. Then, a $\delta$-length time window starting at time $t \in \mathcal{T}$ is defined as

\begin{linenomath}\begin{equation}
w^{\delta}_t = [t, t+\delta - 1] \cap \mathbb  N.
\end{equation}\end{linenomath}
For a given time interval starting at time $T_s$ and finishing at time $T_f > T_s$ such that $(T_f - T_s) \geq \delta $ over which a signal of choice is studied, the set of all $\delta$ time windows that can be extracted is further expressed as

\begin{linenomath}\begin{equation}
\mathcal W_{\delta} = \left\{ w^{\delta}_t | t \in \{T_s, \ldots, T_f-\delta \} \right\}.
\end{equation}\end{linenomath}
Note that $\mathcal W_{\delta}$ is a set of sets of integers.

\subsubsection*{Time series of resource signals}
\label{subsub:measurement_time_series}

Let $\mathcal{R}$ be the set of all available types of intermittent renewable resources, e.g. wind or solar. It is assumed that a single resource $r \in \mathcal R$ can be harvested at each location $l \in \mathcal L$, e.g. either wind or solar but not both simultaneously, which allows for the definition of a surjective mapping $q : \mathcal L \mapsto \mathcal R$ associating its resource $r=q(l)$ to any location $l$. It is worth noting that the mapping $q$ can be defined to associate different resource types to different locations. Then, a signal $ \mathbf{s}^{l,q(l)} \in \mathbb R_{\geq 0}^T$ is defined as a time series representing a renewable resource $r=q(l)$ available at a specific location $l \in \mathcal L$, e.g. wind speed or solar irradiance values. An element of this time series is denoted as $s_t^{l, q(l)}, \mbox{ } \forall t \in \mathcal{T}$. A similar convention will be employed to denote elements of other time series introduced in the remainder of this section.

\subsubsection*{Capacity factor time series and vectors}
\label{subsub:load_factor_time_series_and_vectors}

For any resource $r \in \mathcal{R}$, let $\mathcal{I}^{r}$ be the set of technologies available to harvest it, e.g. various turbine models for wind energy conversion. Then, let $p : \mathcal L \mapsto \bigcup_{r \in \mathcal R} \mathcal I^{r}$ be a mapping associating a conversion technology to each location $l \in \mathcal L$. Accordingly, a transfer function $F^{p(l)}: \mathbb{R}_{\ge 0}^T \mapsto [0, 1]^T$ mapping the raw renewable resource signal $\mathbf{s}^{l, q(l)}$ to a capacity factor time series $\mathbf{u}^{l}$ can be defined, such that the following component-wise relationship holds

\begin{linenomath}\begin{equation}
u_t^{l} = F^{p(l)}\big(s_t^{l,q(l)}\big), \mbox{ } \forall t \in \mathcal{T}.
\end{equation}\end{linenomath}
We also introduce a mapping $\hat{U}_\delta: \mathcal{W}_\delta \times [0, 1]^T \mapsto [0, 1]^\delta $ returning the $\delta$-length truncation of a given capacity factor time series $\mathbf{u}^{l}$ over a time window $w^\delta_t$, such that

\begin{linenomath}\begin{equation} 
\hat{U}_{\delta} \Big({w^{\delta}_t}, \mathbf{u}^{l} \Big) = \left[ u^{l}_t , u^{l}_{t+1}, \ldots, u^{l}_{t+\delta-1} \right].
\end{equation}\end{linenomath}

\subsubsection*{Measure mappings}
\label{subsub:measure_mappings}

For any time window length $\delta$, we define a mapping $N_\delta : [0, 1]^\delta \mapsto [0, 1]$ associating a scalar between zero and one to any vector in the $\delta$-dimensional hypercube $\mathcal{C} = [0, 1]^\delta$. When dealing with capacity factor time series, standard statistical measures like the mean, median or percentiles can be easily accommodated in this framework. In what follows, this mapping will be taken to compute the mean of the elements of an input vector

\begin{linenomath}\begin{equation}
N_{\delta}({\bold v}) = \frac{1}{\delta}  \sum_{i=1}^{\delta} v_i, \mbox{  } \forall {\bold v} \in \mathcal{C}.
\end{equation}\end{linenomath}
\subsubsection*{Critical time windows}
\label{subsub:critical_time_windows}

Let $\alpha \in [0,1]$ be a capacity factor threshold. For a given location $l \in \mathcal L$ and associated capacity factor time series ${\bold  u}^{l}$, we denote by $\Omega^{l}_{\delta\alpha}$ the set of all $(\delta, \alpha)$-critical time windows, such that

\begin{linenomath}\begin{equation}
\Omega^{l}_{\delta\alpha} = \left\{  w_t^\delta \bigg|  w_t^\delta \in \mathcal W_{\delta},   N_{\delta}\bigg( \hat{U}_\delta \Big(w_t^\delta, {\bold  u}^{l} \Big)\bigg)   \leq \alpha \right\}.
\end{equation}\end{linenomath}
Put simply, the set of $(\delta, \alpha)$-critical time windows gathers all $\delta$-length time windows having an average capacity factor value smaller than $\alpha$.

\subsubsection*{Critical locations for a given time window}
\label{subsub:proportion_mapping}

Let $L \in \mathcal P(\mathcal L)$ be a subset of locations, $\delta$ and $\alpha$ an input time window length and a capacity factor threshold, respectively. Using the set of $(\delta, \alpha)$-critical time windows $\Omega^{l}_{\delta\alpha}$ introduced earlier, we define a family of mappings  $\Pi_{\delta\alpha}: \mathcal P(\mathcal L) \times \mathcal W_\delta \mapsto [0,1]$ parametrised by the doublet ($\delta, \alpha$) and returning the proportion of locations in $L$ for which a given time window $w_t^{\delta}$ is $(\delta, \alpha)$-critical. More formally, we have

\begin{linenomath}\begin{equation}
\Pi_{\delta \alpha}(L,w^\delta_t) = \frac{ card \left( \left\{ l \in L | w^\delta_t \in \Omega^{l}_{\delta\alpha}\right\}\right)}{card \left(L\right)},
\end{equation}\end{linenomath}
where $card$ associates its cardinality to any finite set.
\subsubsection*{Spatio-temporal criticality}
\label{subsub:delta_alpha_beta_risk}

Let $\beta \in [ 0, 1 ]$ be a geographical coverage threshold, indicating a proportion of locations within a given geographical area. Based on the previously-defined mappings $\Pi_{\delta \alpha}$, we further introduce a family of mappings $\Gamma_{\delta\alpha\beta}: \mathcal P(\mathcal L) \mapsto [0,1]$ computing the proportion of time windows found to be critical for a proportion of locations greater than $\beta$. In mathematical terms,

\begin{linenomath}\begin{equation}
\Gamma_{\delta\alpha\beta}(L)  = \frac{card \left( \left\{ w_t^\delta \in \mathcal W_\delta | \Pi_{\delta\alpha}(L,w_t^\delta) \geq \beta   \right\}\right)}{ card \left( \mathcal W_\delta\right)}.
\end{equation}\end{linenomath}
This metric will be referred to as the \textit{criticality indicator} or \textit{criticality index} in the sequel. At this stage, it is worth giving more intuition about the physical nature and usefulness of the criticality index, especially in the context of power system planning. 

In physical terms, the geographical coverage threshold $\beta$ defines the proportion of locations which must experience average power production levels below a pre-defined capacity factor threshold value $\alpha$ in order for a time window to be counted as critical. Then, for any instance of the triplet $(\delta, \alpha, \beta)$, the criticality index provides information about the temporal aspect of criticality for a locations set $L$ of choice. For nontrivial instances of the triplet $(\delta, \alpha, \beta)$, when the criticality index has a value close to 1, most time windows are critical for the set of locations considered. This indicates very little complementarity between those locations, and constitutes a poor choice for deploying generation sites. Conversely, if the criticality index has a low value, comparatively few time windows are critical, indicating a set of locations with higher complementarity. In the particular case of $\beta=1.0$, generation levels in all locations have to be below the threshold $\alpha$ to flag the time window as critical across $L$. As $\beta$ decreases, the influence of individual locations in the overall time window criticality assessments expands. For instance, for a $\beta$ value equal to the inverse of the cardinality of a given locations set, if average generation levels below $\alpha$ are recorded at a single location over a given time window, the latter is counted as critical (for the entire system).

An alternative, insightful way of interpreting the value of the criticality index is in terms of empirical probabilities. More precisely, for a given set of locations $L$ and triplet $(\delta, \alpha, \beta)$, the value of the criticality indicator can be understood as the likelihood of obtaining a time window which is $\beta-$critical when drawing a time window uniformly at random from the set of all time windows considered. This interpretation allows for straightforward comparison between different criticality index scores, e.g. if for a given locations set the criticality index value is twice higher than for another set, the probability of encountering critical windows over the former locations set is twice higher.

\subsection{Applications of the Criticality Indicator}
\label{sub:variations_crit_indicator}

In this subsection, a VRE resource complementarity assessment criterion derived from the criticality indicator is proposed, along with a family of optimisation problems identifying sets of locations with maximum complementarity within a given geographical area and over a time period of choice.

\subsubsection{Spatio-temporal criticality analysis}
\label{subsub:indicator_j1}

For specific VRE signal types and associated conversion technologies, successively computing values of the spatio-temporal criticality index $\Gamma_{\delta\alpha\beta}$ for different instances of the triplet $(\delta,\alpha,\beta)$, which will be referred to as control settings in the sequel, enables the direct evaluation of resource quality and complementarity within a spatial domain of interest represented by a set of locations $L \in \mathcal P(\mathcal L)$ and over a time horizon $\mathcal{T}$ of choice.

\subsubsection{Optimal deployment of generation sites}
\label{subsub:indicator_j3}
In order to identify the optimal deployment pattern of $n$ generation sites within a geographical area $L \in \mathcal{L}$, an optimisation problem can be defined in which the spatio-temporal criticality indicator is minimised over all subsets of locations with cardinality $n$ and possibly satisfying constraints enforcing deployment in $k$ arbitrary sub-regions. The optimal deployment pattern is then simply obtained as the argument of the minimum. In the case at hand, optimality can be interpreted as the best level of continuity in the power supply any feasible set of locations can offer.

More formally, let $k \in \mathbb{N}$ denote the number of subsets into which the locations set $L$ considered is partitioned. Moreover, let $\mathcal{P}^k( \mathcal L) = \mathcal{P}( \mathcal L) \times ... \times \mathcal{P}( \mathcal L)$ be the set formed by the Cartesian product of $k$ power sets $\mathcal{P}(\mathcal{L})$, and let $\mathbb{N}^k = \mathbb{N} \times ... \times \mathbb{N}$ be the set of k-tuples of nonnegative integers. Then, let $C_L = \{L_1, ..., L_k\} \in \mathcal{P}^k(\mathcal{L})$ be an arbitrary collection of subsets covering $L$ such that $L = \bigcup_{i=1}^k L_i$ and $L_{i} \cap L_{j} = \emptyset, \mbox{ } \forall i \ne j$. Finally, let $C_n = \{n_1, ..., n_k\} \in \mathbb{N}^k$ be a collection of nonnegative integers, with $n_i$ indicating the number of generation sites to be selected in sub-region $L_i$. The total number $n$ of sites to be deployed across $L$ is thus recovered as $\sum_{i=1}^k n_i = n$. At this point, for any instance of the control settings $(\delta, \alpha, \beta$), the problem of minimising the criticality indicator over all subsets of cardinality $n$ whilst satisfying geographical deployment constraints can be introduced

\begin{linenomath}\begin{align}
\underset{\{L'_{1}, ..., L'_{k}\}} \min \hspace{10pt} &\Gamma_{\delta \alpha \beta}\bigg(\bigcup_{i=1}^{k} L^{'}_{i}\bigg)\\
\mbox{s.t. } & L'_{i} \subset L_i, \mbox{ } i \in \{1, ..., k\}\nonumber\\
& card \left(L'_{i}\right) = n_i, \mbox{ } i \in \{1, ..., k\}\nonumber.
\end{align}\end{linenomath}
In particular, for $k=1$, this optimisation problem reduces to the selection of a single subset of locations of pre-specified cardinality within a starting set $L$ in order to minimise the criticality indicator. Different instances of this optimisation problem can be generated for distinct instances of the parameters $C_L$ and $C_n$. Hence, the optimisation problem can be used to define a convenient mapping $\Gamma_{\delta \alpha \beta}^{min}: \mathcal{P}^k(\mathcal{L}) \times \mathbb{N}^k \mapsto [0, 1]$ returning the minimum value $\Gamma_{\delta \alpha \beta}^{min}(C_L, C_n)$ of the criticality index for any two input parameters $C_L$ and $C_n$, which greatly simplifies notation in future sections.

Taking the argument of the minimum in the above optimisation problem allows to derive the set of locations minimising the criticality index for a given partition of the geographical and deployment constraints as described by the parameters $C_L$ and $C_n$. More precisely, in the case of a single global minimum, taking the argument of the minimum will yield a collection of sets of locations, the union of which directly gives the optimal deployment pattern.
As some natural symmetry may exist in the data, there may not be a single collection of locations sets minimising the criticality indicator. In this context, taking the argument of the minimum returns a set of sets of locations sets, and the most suitable one can for instance be selected by visual inspection on a map. In any case, these technical details become irrelevant when implementing and solving the optimisation problems.

\section{Test Case}
\label{sec:applications}
In this section, a particular application is proposed in order to illustrate the usefulness of the indicator defined in Section \ref{sec:formalisation}. More specifically, a wind regime analysis is conducted in South Greenland and France in the context of an electrical interconnection between Greenland and mainland Europe, as part of a broader global grid. At first, a general assessment of the spatio-temporal criticality observed in wind signals across (i) Greenland, (ii) France and (iii) within their spatial aggregation is carried out via the $\Gamma_{\delta \alpha \beta}$ index. The value of the capacity factor threshold ($\alpha$) is fixed, and the impact of different window lengths ($\delta$) and geographical coverage threshold ($\beta$) values is assessed. Then, the optimal distribution of wind generation sites within those geographical areas is analysed under various deployment constraints.

\subsection{Data Acquisition}
\label{sec:data}

This subsection introduces the wind signal datasets used in the current work and briefly discusses the selection of geographical regions employed in Section \ref{sec:results} analyses. Two distinct data sources are used in the wind signal acquisition process for the geographical regions under consideration, the ERA5 reanalysis and the MAR (\textit{Modèle Atmosphérique Régional}) climate models. Both these sources provide historical wind time series sampled at hourly resolution and covering the last ten years, i.e. 2008-2017. 

Resource data within the boundaries of France is acquired via a climate reanalysis model. That is, a process in which in situ and satellite measurements are interpolated in space and time using numerical weather prediction models, yielding a consistent estimation of meteorological observations over structured, regular grids \cite{gelaro2017modern}. Some of the features that have contributed to the success of such methods in energy applications include their extended temporal coverage (e.g., decades), high temporal resolution (e.g., hourly) and uniform sampling, provision of many meteorological variables (e.g., wind speed and direction, air pressure and humidity, solar irradiance components, etc.) or their free-of-charge availability \cite{rose2015reanalysiswind}. In the context of wind power generation studies, the shortcomings of these reanalysis datasets stem mainly from the relatively coarse spatial resolution of their underlying grids. In complex terrain conditions, this latter aspect limits the accurate replication of local, topography-induced winds patterns \cite{olausson2018eramerra, rose2015reanalysiswind, staffell2016merra, pfenninger2016merra}.

For this study, the reanalysis source used is the state-of-the-art ERA5 dataset, developed by the \textit{European Centre for Medium-Range Weather Forecasts} (ECMWF) through the \textit{Copernicus Climate Change Service} (C3S). Upon complete development of the model, it will include atmospheric parameters at various temporal resolutions (i.e., down to hourly) spanning from 1950 to near real time. Set on a geodesic grid (a grid across the surface of the Earth), data is provided at a spatial resolution of $0.28^{\circ}\times0.28^{\circ}$ \cite{ecmwf2018era5}. This aspect makes spacing over the globe different for parallels (across latitudes) and meridians (across longitudes). On the one hand, parallels are spaced at regular intervals and this implies that the aforementioned resolution in degrees corresponds to a longitudinal distance between nodes equivalent to a constant value of \SI{31}{km}. On the other hand, the convergence of meridians at the poles leads to a node spacing across parallels that varies with latitude. In the particular case of ERA5, this spacing ranges from \SI{0.15}{km} close to the poles to \SI{31}{km} at the equator. For the purpose of the following analysis, hourly-sampled wind time series above continental France are retrieved for an altitude of 100 meters above ground level.

A second approach and source of data, which alleviate some of the previously mentioned limitations of reanalysis models, is used for wind data acquisition over Greenland. To be more specific, an atmospheric model is used to provide an accurate replication of weather phenomena up to the mesoscale level (i.e., usually on a spatial scale of tens to hundreds of kilometers). An atmospheric model is a mathematical representation of the dynamics of atmospheric phenomena, usually based on a set of partial differential equations, solved over limited geographical scopes and time durations via numerical integration schemes. The boundary conditions of the spatial domain usually come from reanalysis data over the same time frame, a feature that enables the preservation of synoptic scale (i.e., in the order of thousands of kilometers) effects that drive the dynamics of lower-level, mesoscale systems.

Within the current work, the regional MAR model is used. Repeatedly validated over Greenland, MAR has been specifically developed for simulating atmospheric conditions over polar regions \cite{fettweis2017reconstructions}. The main advantage of this tool lies in its ability to accurately reproduce, at refined spatial resolution (down to \SI{5}{km}$\times$\SI{5}{km}), particular features of the atmospheric circulation over Greenland, including the local, semi-permanent katabatic flows that may enable high levels of wind power generation. For this study, ERA5 fields (e.g., wind speeds, air temperature, humidity or pressure) are used as forcing at the spatial boundaries above South Greenland to retrieve MAR-based hourly wind time series at 100 meters above ground level.

\begin{figure}[!b]
	\centering
	\includegraphics[width=\columnwidth]{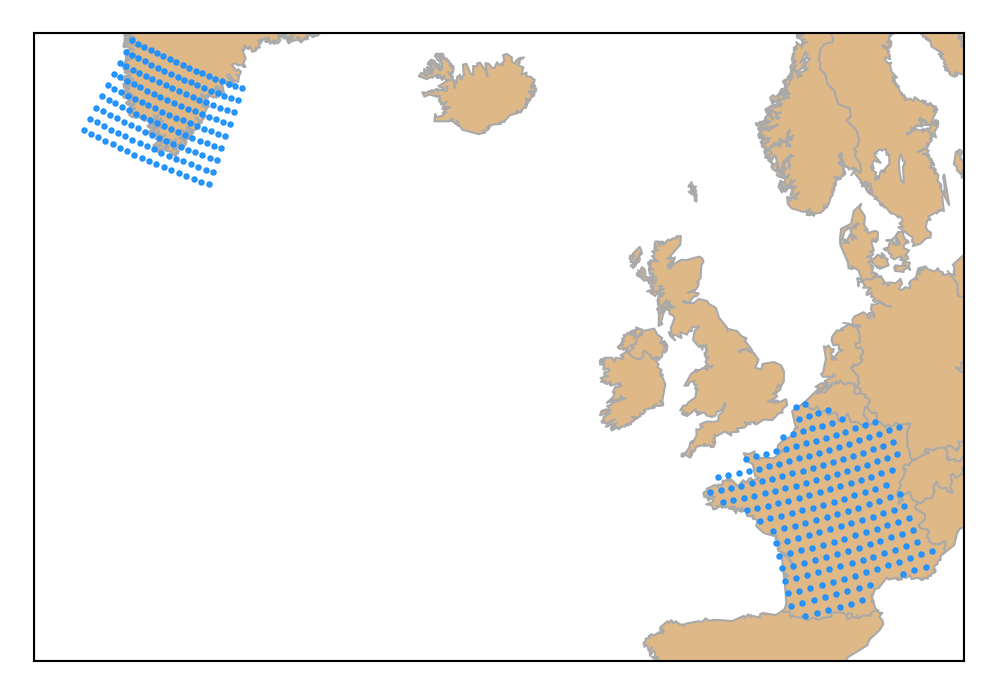}
	\caption{Locations considered for the wind resource assessment in France and South Greenland. The upper-left corner of the map displays the geographical points included in $\mathcal L_{G}$, while the locations superimposed over the French territory comprise $\mathcal L_{F}$.}
	\label{fig:map_data}
\end{figure}

The selection of South Greenland as the focus of this analysis is done via a MAR-based visual inspection of the entire Greenlandic land mass that reveals vast wind resource, relatively high temperatures and favourable topography for the sub-region considered \cite{radu2018greenland}. The locations sets used in the following analysis comprise geographical points which exist in both the ERA5 and MAR grids, respectively. In the following, the set of all locations in France is denoted by $\mathcal L_{F}$, while sites in South Greenland are grouped in $\mathcal L_{G}$. A third set, denoted by $\mathcal L_{FG}$, is defined as the union of the two location sets $\mathcal L_{FG} = \mathcal L_F \cup \mathcal L_G \ .$ All the aforementioned locations sets are depicted in Figure \ref{fig:map_data}.

\subsection{Defining the Conversion Technology}
\label{sec:transfer_function}

\begin{figure}[!t]
	\centering
	\includegraphics[width=\linewidth]{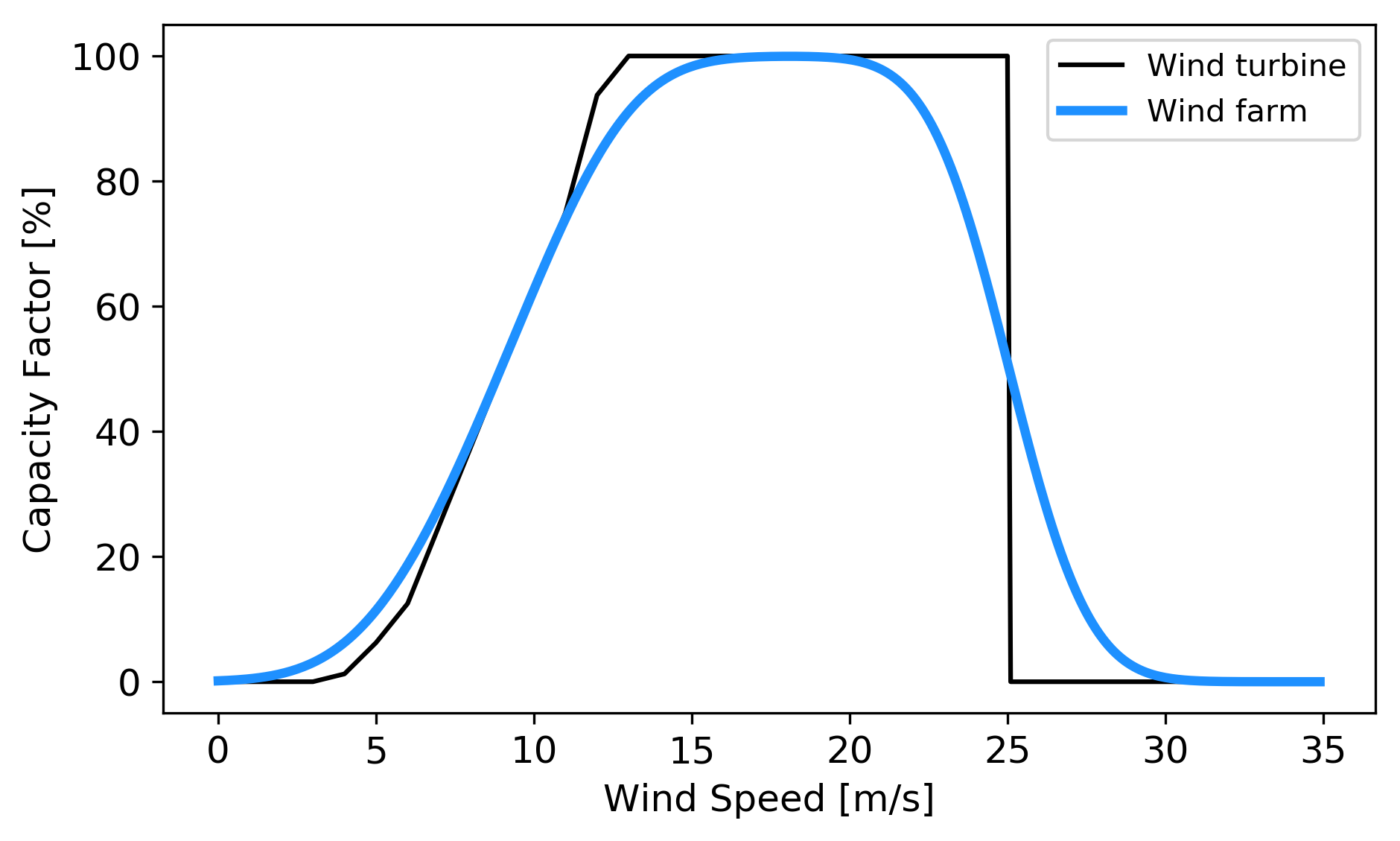}
	\caption{Single turbine and wind farm transfer functions associated with aerodyn SCD 8.0/168 units \cite{aerodyn_scd_8MW}.}\label{fig:transfer_function_curve}
\end{figure}

The mapping of wind speeds to hourly average capacity factors is performed via a transfer function mimicking the normalised power output of a wind farm comprising identical units which are geographically distributed in the direct vicinity of the site of interest. This approach is consistent with power system planning processes, where one is usually interested in developing wind farms rather than constructing a single turbine. As proposed in \cite{staffell2016merra}, such a transfer function is derived from the power curve of a representative wind turbine (the \textit{aerodyn SCD 8.0/168}, in this particular case) by means of a Gaussian fit, where a 100\% availability of the individual units is assumed (i.e., no down times due to maintenance, icing etc.). The result of this approach is depicted in Figure \ref{fig:transfer_function_curve}. In addition, this specific wind energy converter is selected for illustration purposes only, regardless of its appropriateness for deployment at the locations considered.

\section{Results}
\label{sec:results}
This section presents a detailed discussion of results generated applications of the criticality index introduced in Section \ref{sub:variations_crit_indicator}. We start by examining the values taken on by the criticality indicator for various instances of the ($\delta, \alpha, \beta$) triplet within the selected geographical areas. Then, the optimal deployment of generation sites in different geographical set-ups is discussed. Lastly, the current framework is leveraged to highlight potentially undesirable consequences of current power systems planning practices, which primarily favour electricity generation potential and disregard the complementarity in electricity production regimes when selecting wind farm deployment sites.

\subsection{Spatio-temporal Criticality Assessment}
\label{sec:indicator_r0}

Initially, a comprehensive evaluation of wind resource complementarity in the available locations sets over a time period stretching the last ten years (2008-2017) is conducted using the $\Gamma_{\delta \alpha \beta}$ index defined in Section \ref{subsub:indicator_j1}. As a reminder, this indicator provides information regarding the proportion of time windows, within the aforementioned period, during which wind generation is below a certain threshold ($\alpha$) relative to the associated installed capacity, for various time window lengths ($\delta$ - taking values from one hour to ten days), geographical coverage proportions ($\beta$ - from 0.5 to 1.0, in this example) and for different geographical areas (e.g., $\mathcal L_F, \mathcal L_G, \mathcal L_{FG}$). In the following example, as well as in all other applications subsequently proposed, the generation level threshold ($\alpha$) is set to 35\%. This choice stems from it standing between the average capacity factors of the two regions, as computed from available data (i.e., 22\% for all French sites and 48\% for the locations in Greenland). The process of mapping wind signals to normalised power output values is done using the methodology introduced in Section \ref{sec:transfer_function}. Formally, we compute $\Gamma_{\delta \alpha \beta}$ with a capacity factor threshold ($\alpha$) of 0.35 for all locations sets in $\{\mathcal L_F, \mathcal L_G, \mathcal L_{FG}\}$, for all window lengths in $\{1, 24, 72, 168\}$ and for all geographical coverage thresholds in $\{0.5, 0.6,\ldots, 1.0\}$. The associated results are depicted in Figure \ref{fig:r0}.

\begin{figure}[!b]
	\centering
	\includegraphics[width=\linewidth]{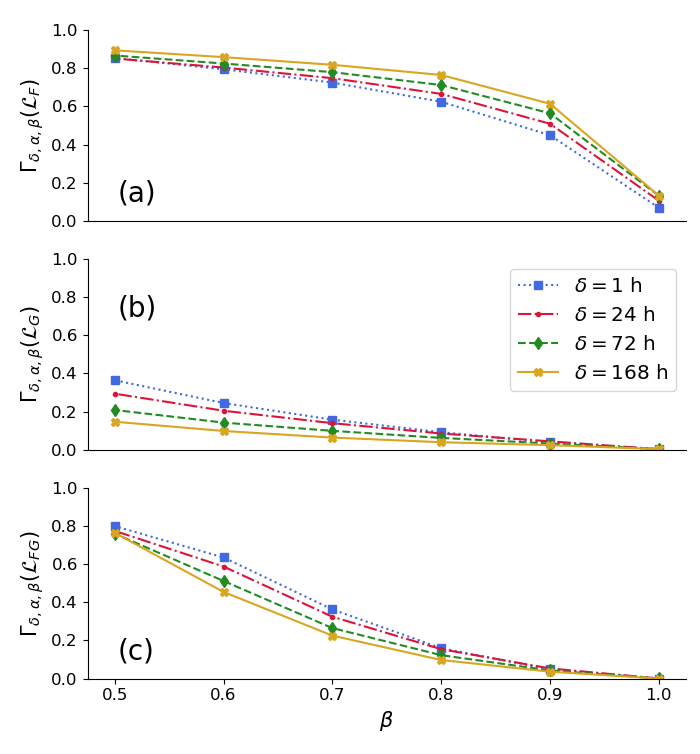}
	\caption{Criticality indicator $\Gamma_{\delta \alpha \beta}$ values for a capacity factor threshold (i.e., $\alpha$) of 35\% for all locations in (a) France ($\mathcal L_F$), (b) Greenland ($\mathcal L_G$) and (c) the aggregation of the two ($\mathcal L_{FG}$).}
	\label{fig:r0}
\end{figure}

First, as can be observed in all subplots of Figure \ref{fig:r0}, the value of the $\Gamma_{\delta \alpha \beta}$ indicator decreases as the geographical coverage threshold ($\beta$) increases, regardless of the time window length ($\delta$) considered. This is happening since time window criticality is easier to achieve at lower $\beta$ values, where less locations are required to have insufficient generation in order for the time window to be critical (see detailed explanation in Section \ref{sub:math_formalization}). In addition, for all locations sets observed, it appears that considering all available generation sites (i.e., $\beta = 1.0$) in the computation of $\Gamma_{\delta \alpha \beta}$ leads to values of the criticality indicator close to zero, an aspect that highlights significant wind regime diversity even at regional scale.

\begin{table}[!t]
	\centering
	\renewcommand{\arraystretch}{1.2}
	\caption{Percent values of the criticality index $\Gamma_{\delta \alpha \beta}$ for various geographical threshold values ($\beta$) and for all considered locations sets, considering a capacity factor threshold ($\alpha$) of 35\% and a window length ($\delta$) of 24 hours.}
	\label{tab:r0}
	\begin{tabular}{c|cccccc}
		\diagbox[width=1.1cm, height=1cm]{$L$}{$\beta$}& 0.5 & 0.6 & 0.7 & 0.8 & 0.9 & 1.0\\ \hline
		$\mathcal L_F$ & 85.1 & 80.4 & 74.8 & 66.6 & 50.9 & 10.6 \\
		$\mathcal L_G$ & 29.4 & 20.4 & 13.9 & 8.5 & 4.4 & 0.3 \\
		$\mathcal L_{FG}$ & 77.4 & 58.8 & 32.4 & 15.4 & 5.3 & 0.1 \\
	\end{tabular}
\end{table}

The change in criticality index values associated with the locations in France ($\mathcal L_F$) for different $(\delta, \alpha, \beta)$ instances is depicted in Figure \ref{fig:r0}a. For a time window length of 24 hours, as the geographical coverage threshold increases from 0.5 to 1.0, the proportion of critical windows drops from 85.1\% ($\beta = 0.5$) to 66.6\% ($\beta = 0.8$) and, finally, to 10.6\% if all locations are simultaneously assessed ($\beta = 1.0$). The same decreasing trend applies for all other $\delta$ values, with the score of the $\Gamma_{\delta \alpha \beta}$ index observed to be increasing as the window length is extended. Results corresponding to the available locations in Greenland are presented in Figure \ref{fig:r0}b. Superior resource quality in this region is evidenced by the range of output values of the criticality index. They show, in case of a $\delta$ value of 24 hours, a drop in the proportion of critical time windows from 29.4\% (for $\beta = 0.5$) to 8.5\% ($\beta=0.8$) and 0.3\%, when taking all available locations into account ($\beta = 1.0$). Contrary to the change in criticality index values for the locations in France, we observe that, for each geographical coverage threshold ($\beta$), the criticality index values in Greenland decrease as the length of the time window increases. 

In this particular ($\delta, \alpha, \beta$) configuration, the opposing development of $\Gamma_{\delta \alpha \beta}$ with respect to $\delta$ stems from (i) the positioning of $\alpha$ between the average capacity factors of the two regions and (ii) the utilisation of a measure mapping returning the mean signal over each time window. Under these assumptions, we observe the following. On the one hand, the chosen $\alpha$ is greater than the average capacity factor observed in France (i.e., 22\%). Given the measure mapping considered, the probability of a random time window sample being critical is greater for large values of $\delta$ since less frequent occurrences of high energy output in this region are averaged over the time window length (thus cancelling the extreme wind events and rendering the time window critical overall). On the other hand, $\alpha$ is smaller than the estimated average capacity factor in Greenland (i.e., 48\%). In this case, the opposite of the aforementioned situation holds, with short time windows being more likely to turn out critical for a given ($\alpha, \beta$) pair. This happens since, for a resource-rich area, low wind events have a greater impact (that is, time window criticality) on short time horizons, while the effects of those same events are often annihilated when computing the measure mapping for larger $\delta$ values.

Lastly, the outcome of coupling the two regions ($\mathcal L_F$ and $\mathcal L_G$) is displayed in Figure \ref{fig:r0}c. As for Greenland, the criticality index decreases as the time window length ($\delta$) increases. It should also be noted that, except for the $\beta=1.0$ case, the criticality indices in Figure \ref{fig:r0}c are, for all $(\delta, \beta)$ configurations, smaller than the ones associated with the locations in France ($\mathcal L_F$) and greater than the ones corresponding to the sites in Greenland ($\mathcal L_G$). This can be explained by the fact that the influence of the inferior wind resource in France is more pronounced for lower values of the geographical coverage threshold. Nevertheless, the impact of the high-quality wind resource of Greenland is observed in the shape of the plotted curves which, compared to the ones associated to France (Figure \ref{fig:r0}a), change curvature, leading to a steeper drop of the $\Gamma_{\delta \alpha \beta}$ values as the $\beta$ factor increases. Numerically, given a window length of 24 hours, the criticality index decreases from 77.4\% ($\beta=0.5$) to 15.4\% ($\beta=0.8$) and, considering all locations ($\beta=1.0$), to 0.1\%. These observations already give a clear indication of the benefits of harvesting wind energy in Greenland in order to complement the existing wind regimes in France. For the sake of clarity, the output values of the criticality indicator for different locations sets, geographical coverage proportions and considering a 24-hour time-window length are summarised in Table \ref{tab:r0}.

\subsection{Optimal Deployment of Generation Sites}
\label{sec:optimisation_rmin}

This subsection details the results of the minimisation problem defined in Section \ref{subsub:indicator_j3}. More specifically, the optimal deployment of $n$ generation sites across $k$ areas is assessed for different input regions (e.g., $\mathcal L_F, \mathcal L_G, \mathcal L_{FG}$), a given triplet $(\delta, \alpha, \beta)$ and taking account of pre-defined constraints on the geographical repartition of wind farms throughout the $k$ sub-regions. A time window length $\delta$ of 168 hours (one week), a capacity factor threshold $\alpha$ of 35\% and a geographical coverage threshold $\beta$ of 1.0 are considered for illustrative purposes in the following example. A heuristic, which is described in detail in the appendix, is proposed to provide an approximate solution to the optimisation problems at hand, and results are shown in Figure \ref{fig:rmin}.

\begin{figure}[!t]
	\centering
	\includegraphics[width=\linewidth]{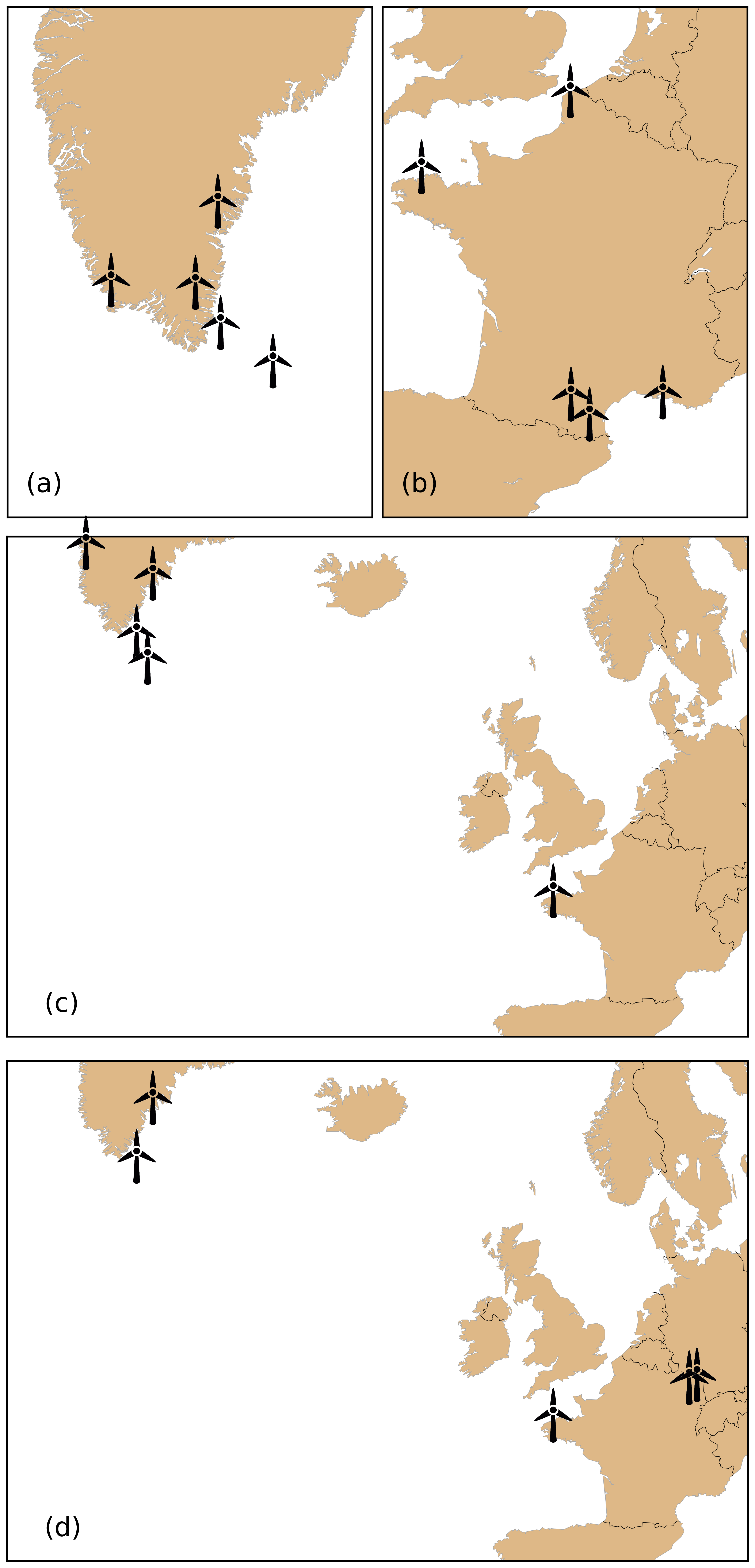}
	\caption{Visualisation of the optimal deployment of five wind farms in (a) Greenland, (b) France, (c) the aggregation of the two without and (d) with constraints on the geographical repartition. Results corresponding to a time window length ($\delta$) of one week, a capacity factor threshold ($\alpha$) of 35\% and a geographical coverage threshold ($\beta$) of 1.0.}
	\label{fig:rmin}
\end{figure}

Optimal deployment of five generation sites within the available locations in southern Greenland and continental France is shown in Figure 4a and Figure 4b, respectively. At first glance, it can be observed that the identified generation sites are evenly distributed over the regions of interest, again revealing (this time visually) the complementarity of wind regimes across dispersed locations even on the regional scale. Also, the actual wind farm locations in both cases can be explained via well-documented, prevailing local or regional wind regimes. For Greenland, the deployment of all but one wind farm is in line with the spatial occurrence of local katabatic flows \cite{radu2018greenland}. Shifting to France, two wind farms are deployed in the resource-richer north (as will be detailed in the following section), while the remaining ones are built in southern areas often swept by strong, local winds (the Mistral and Tramontane) \cite{obermann2016windsfrance}. Numerically, the five locations in France display a 15\% probability of critical window occurrence, while the superior wind resource in Greenland translates to an almost zero (0.4\%) criticality index value.

The results of two variations on the same minimisation problem applied to the aggregated locations set ($\mathcal L_{FG}$) are discussed next. Figure 4c displays the optimal distribution of generation sites in France and Greenland solely based on resource quality, without any additional deployment constraints. This variant corresponds to the case of a single, aggregated input region (i.e., $k=1$). Numerically, the proportion of one-week-long time windows with generation levels under 35\% capacity factor across all five sites is only 0.005\% over the full time horizon. Nevertheless, constraints on the geographical repartition of the generation sites can play a decisive role, as can be observed in Figure 4d. This plot shows the optimal deployment of the five sites in case the number of wind farms to be developed in Greenland is limited to two (an example of the more general problem formulated in Section \ref{subsub:indicator_j3}). In this case, the criticality index score increases to 0.06\%.

\subsection{Comparison to Average Capacity Factors as Primary Criterion}
\label{sec:optimisation_tradeoff}

It is particularly insightful to compare the generation site selection according to different criteria. In particular, the usual criterion used for wind farm deployment is the generation potential (or the average capacity factor) of the site. A comparison between this indicator and the complementarity criterion proposed in this paper is presented in Figure \ref{fig:subopt} for the case of France. In this plot, the black markers are associated with the solution of the minimisation problem, as presented in Section \ref{sec:optimisation_rmin}. The blue ones correspond to the five best locations in $\mathcal L_F$ (see Figure \ref{fig:map_data}), strictly from a generation potential perspective (i.e., locations are ranked based on integrated capacity factor values over the available time horizon). The point displayed in both colours represents a generation site common to both solutions.

Numerically, as detailed in Section \ref{sec:optimisation_rmin}, the optimisation problem returns a criticality index value of 15\%. The criticality index associated with the five most productive sites is higher (23.6\%). In other words, compared to the locations set identified through the minimisation problem, the likelihood of recording a 168-hour long critical window across the five most productive sites is 60\% higher (for a capacity factor threshold of 35\%). This increase stems from the geographical proximity of the locations, which makes them subject to very similar wind regimes. However, the improvement in criticality index value for the locations with highest complementarity (i.e., having increased continuity of supply) comes at the expense of very high capacity factor values. While the five most productive wind farms have an aggregated capacity factor value of 46\% over the entire time domain considered, the locations with highest complementarity only boast an aggregated capacity factor value of 34\%.

\begin{figure}[!t]
	\centering
	\includegraphics[width=\linewidth]{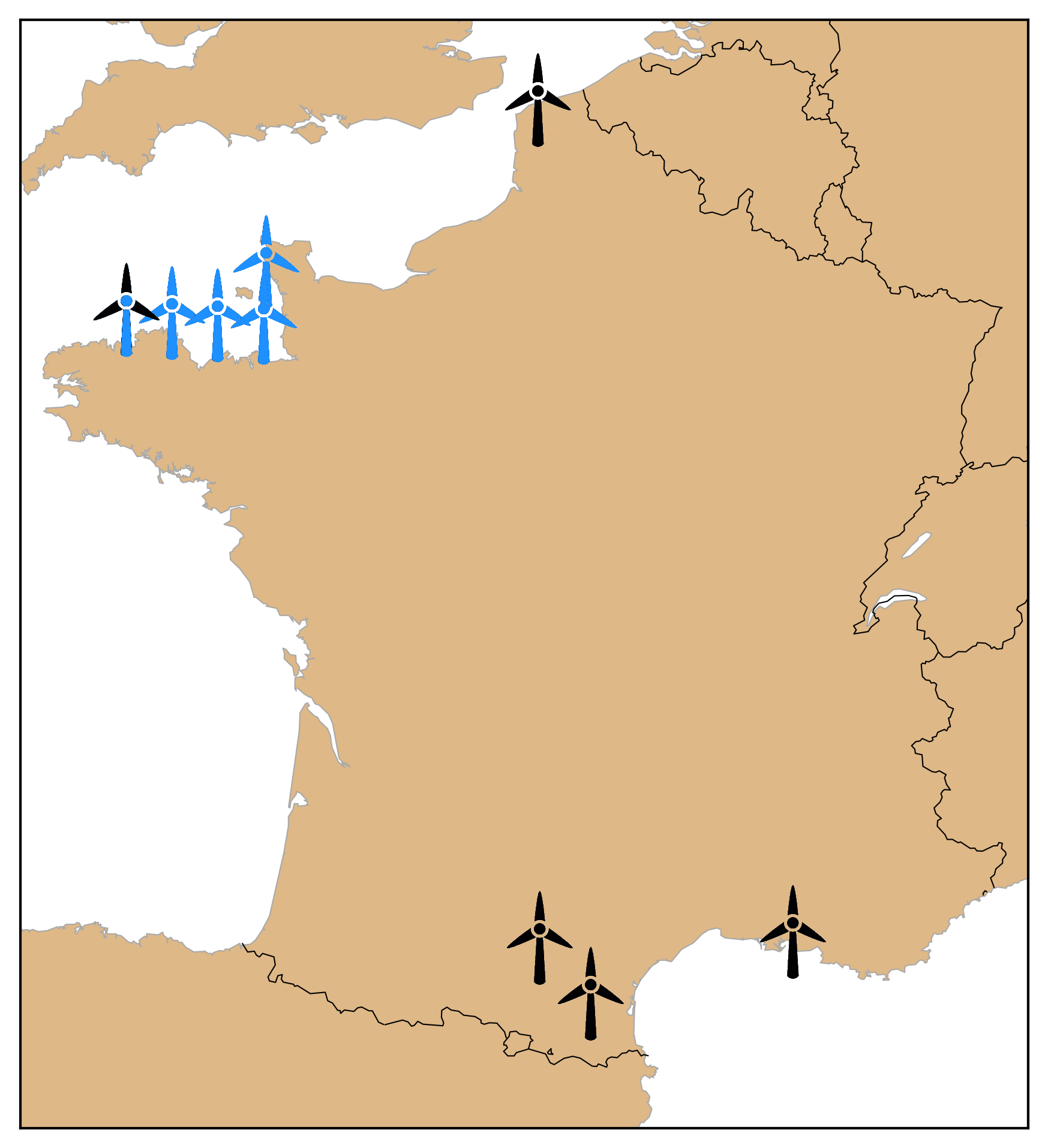}
	\caption{Comparison of criticality index-based (black) and electricity yield-based (\textcolor{NavyBlue}{blue}) deployment of wind farms within a given subset of locations, based on the criticality indicator values. Example depicting the results for the subset of potential generation sites in France, $\mathcal L_F$.}
	\label{fig:subopt}
\end{figure}

Therefore, such an analysis opens the door to planning decisions based on a trade-off between high production levels and enhanced continuity of supply. According to the particularities and requirements of each individual system, one may opt for extensive generation volumes accompanied by bigger time frames of production under certain levels or for lower generation volumes that, for most of the time, behave as base load generation supplied entirely by renewable energy technologies.

\subsection{Results Discussion}
\label{sec:disc}

A key factor impacting the accuracy of results presented in Section \ref{sec:results} is the accuracy of the raw data used throughout the analysis. In this respect, several studies have identified significant spatial bias attributed to the ability of reanalysis models to recreate VRE resource patterns in certain topographies \cite{rose2015reanalysiswind, staffell2016merra, pfenninger2016merra}. This undesirable feature is also observed in the current study that uses ERA5 as reanalysis database. Comparing modelled average capacity factors (as computed from ERA5 resource data) with actual realisations of the regions where the generation sites in Figure \ref{fig:subopt} are located, non-negligible differences can be observed. On the one hand, for the locations in Southern France, the computed average capacity factor is close to 25\%, a value which is just about the official statistics (27.4\% in 2016) \cite{rte2017bilan}. On the other hand, a clear positive bias of the reanalysis model can be observed for the six locations in Northern France. In their case, a modelled average capacity factor of 45\% is significantly higher than the reported 19 to 22\% values associated with the corresponding French regions in 2016 \cite{rte2017bilan}. The sole purpose of the wind database used in this work is to illustrate the mathematical framework proposed in Section \ref{sec:formalisation} and the continuous improvement of reanalysis models over the years suggests that these tools could be used more reliably in such applications in the near future. Nevertheless, final decisions on wind farm site selection should be confirmed by extensive in situ measurements.

One interesting implication of the proposed methodology stems from the results presented in Section \ref{sec:optimisation_tradeoff}, which suggest the existence of a trade-off between maximising energy volumes (and revenues, in current regulatory settings) and maximising continuity of supply (via the criticality index) in the planning process of VRE generation. This aspect brings into question the regulatory frameworks favouring the former option and potential enhancements that can be addressed in this regard. In a regulatory set-up that incentivises producers to provide ancillary network services (e.g., contributing to a base load provision of electricity from renewable sources), novel business cases could emerge by optimally deploying (from a criticality indicator standpoint) production sites in power systems relying heavily on renewable energy generation.

\begin{table}[t]
	\centering
	\renewcommand{\arraystretch}{1.5}
	\caption{Influence of the capacity factor threshold ($\alpha$) on the trade-off between criticality index gain and average capacity factor loss in the context of deployment strategy comparison between (i) criticality index minimisation and (ii) average capacity factor as selection criterion. Numerical results displayed for one-week-long time windows ($\delta=168$) and under full geographical coverage ($\beta=1.0$).}
	\label{tab:alphas}
	\begin{tabular}{lc|ccccc}
		$\alpha$ & - & 0.15   & 0.35   & 0.40   & 0.45 & 0.65   \\
		$\Delta\Gamma_{\delta \alpha \beta}$& \% & 99.6 & 36.4 & 27.6 & 20.4 & 1.9 \\
		$\Delta\hat{\mathbf{u}}$ & \% & 28.2 & 26.0 & 26.0 & 26.0 & 16.6
	\end{tabular}
\end{table}

The influence of $\alpha$ on the aforementioned trade-off reveals further interesting findings. In this respect, Table \ref{tab:alphas} summarises the criticality index gain ($\Delta\Gamma_{\delta \alpha \beta}$) and the capacity factor loss ($\Delta\hat{\mathbf{u}}$) associated with the two deployment strategies identified in Section \ref{sec:optimisation_tradeoff}, for various capacity factor threshold ($\alpha$) values, a time window length ($\delta$) of 168 hours and full geographical coverage ($\beta=1.0$). Both these measures are defined in relative terms, comparing the results of the criticality index-based deployment strategy with the ones corresponding to the traditional, energy-based approach (blue markers in Figure \ref{fig:subopt}). In this case, for a given ($\delta, \alpha, \beta$) triplet, the criticality index gain represents the decrease in the number of critical windows, while the capacity factor loss quantifies the reduction of energy yield over the entire time horizon. As $\alpha$ varies between 0.15 and 0.65, the numbers show a steep, non-linear decrease of the criticality index gain accompanied by a smoother drop in terms of energy generation potential. This evolution can be explained by interpreting the two extremes. For small values of the capacity factor threshold (e.g., $\alpha=0.15$), the outcome of the minimisation problem includes locations distributed evenly across the region. Such a deployment scheme leads to very good criticality index values (and to substantial gains compared to the criticality indicator associated with the classical scheme), but also to relatively high losses in terms of generation potential (since locations with better complementarity are eligible to be chosen over the ones with very good resource). Distinct results occur for large values of the capacity factor threshold (e.g., 0.65). As $\alpha$ increases, locations with superior yield are selected for wind farm deployment and it often happens that at least part of these generation sites are in relative geographical proximity. Therefore, the results display a fairly low criticality index gain accompanied by a limited decrease in generation potential.

\section{Conclusion and Future Work}
\label{sec:conclusion}
A general methodology to systematically assess the complementarity of renewable resources over both space and time has been presented. The framework relies on the concept of critical time windows, which provide an accurate description of extreme production events in time-domain whilst retaining chronological information, and from which a parametrised family of scalar indicators quantifying complementarity is derived. The so-called criticality indicator constitutes a practical and efficient way of evaluating renewable resource complementarity over arbitrary geographical scopes and time scales. The criticality indicator is further leveraged to propose a family of optimisation problems allowing to identify the deployment pattern of generation sites with smallest occurrence of low production events within a region of interest and under arbitrary spatial distribution constraints.

The versatility and usefulness of the proposed framework is highlighted in a case study investigating the complementarity of wind regimes in and between continental France and southern Greenland. The application of the criticality indicator reveals that a reduction in the occurrence of low system-wide VRE generation events can be achieved when the two areas are spatially aggregated, indicating the potential benefits of such intercontinental electrical interconnections. Moreover, the solutions to optimisation problems derived from the criticality indicator show that the occurrence of low power production events can be substantially reduced on a regional scale by exploiting the diversity in local wind regimes. Likewise, aggregation of generation sites across continents allows both to full take advantage of high-quality resources and reduce the occurrence of low production events for the systems of interest. The relevance of the proposed methodology in a power system planning context is further supported by a comparison of two wind farm deployment strategies (one favouring complementarity and the other based solely on average capacity factors) tested in continental France. These two approaches yield widely-varying outcomes, with implications for planning strategies in future power systems dominated by vast shares of renewable-based generation, where maintaining adequate levels of security of electricity supply usually requires a comprehensive assessment of renewable resource complementarity.

Several research avenues can be pursued in future work. From a methodological perspective, extending the formalism to assign an $\alpha$ to each location rather than taking a uniform value across all locations may yield additional insight into the complementarity of renewable resource signals. The framework could also be applied to investigate the complementarity between different renewable resource types. From a computational standpoint, recasting the proposed optimisation problems in a more generic form would be beneficial, as it would allow to use efficient off-the-shelf solvers, e.g. branch-and-bound, which would provide certificates of optimality. If those efforts prove fruitless, a comprehensive analysis and extension of the proposed heuristic to cases in which $\beta \ne 1$ would be needed. Finally, further exploring the trade-off between maximising complementarity or average capacity factor would be particularly useful for planning purposes. More precisely, quantifying the value of complementarity in economic terms could enable to identify whether transmission, dispatchable generation or storage capacity expansion strategies should be pursued to ensure adequacy in a power system with ever-increasing shares of variable renewable resources. In the same vein, updating the optimisation problems to include other constraints and costs, for instance reflecting a desired level of installed renewable generation capacity or the difficulty to connect to existing infrastructure, would allow for a more complete assessment of renewable generation deployment options.

\appendix
\numberwithin{equation}{section}

\section*{Appendix}
\label{sec:Appendix}
This subsection details the heuristic proposed to solve the optimisation problem introduced in Section \ref{sub:variations_crit_indicator} in the full spatial coverage case in which only time windows critical at all locations are counted when computing the criticality indicator, i.e. $\beta = 1.0$. If the geographical area considered numbers only a few dozen locations, an exhaustive search procedure can be implemented to retrieve the exact solution to the optimisation problem. In case several hundreds of locations or more are considered, such brute force approaches are no longer tractable. Hence, an alternative method trading accuracy for speed and approximating the optimum has to be employed, as described next.

The pseudocode in Algorithm \ref{alg:minalg} captures the spirit of the proposed heuristic. In essence, this heuristic relies on the key insight articulated in the following observation:

\begin{observation}[Property of Criticality Index]
	Let $L$ be a reference set of locations, and $\Gamma = \Gamma_{\delta \alpha \beta}\big(L\big)$ the value of the criticality index for arbitrary values of the $(\delta, \alpha)$-pair, arbitrary set $\mathcal{W}_\delta$, and $\beta = 1.0$. Then, $\Gamma \le \Gamma_{\delta \alpha \beta}\big(L\setminus \{l\} \big), \ \forall l \in L$.
	\label{Observation1}
\end{observation}
\begin{proof}
	Let us assume there exists a location $\hat{l} \in L$ such that $\Gamma_{\delta \alpha \beta}\big(L\setminus \{\hat{l}\} \big) < \Gamma$. This would imply that at least one time window which was critical for $L$ is no longer critical for the proper subset $L\setminus \{\hat{l}\}$. However, by definition of the criticality index for $\beta = 1.0$, time windows must be critical for every location in $L$ to be counted in the first place, so that they must also be critical for all locations in any non-trivial proper subset of $L$, which leads to a contradiction. Hence, no such $\hat{l}$ as previously hypothesised can be found and $\Gamma \le \Gamma_{\delta \alpha \beta}\big(L\setminus \{l\} \big), \ \forall l \in L$.
\end{proof}

In other words, as locations are successively removed from a starting set, a sequence of non-decreasing values of the criticality index is generated. Given an initial set of locations and a desired cardinality for a terminal set, an iterative procedure exploiting this property can therefore be devised to identify a set of locations approaching the minimum of the criticality index (amongst all location sets with the desired cardinality), and also satisfying partitioning constraints. It is worth noting, however, that this procedure may be suboptimal in that it may yield a set of locations not exactly minimising the criticality indicator.

\begin{algorithm}[t]
	\caption{Minimisation problem for $\beta=1.0$}\label{alg:minalg}
	\begin{algorithmic}
		\STATE \textbf{Input:} $k,\  \{(L_1, n_1),..., (L_k, n_k)\}$
	\end{algorithmic}
	\begin{spacing}{1.2}
		\begin{algorithmic}[1]
			\STATE $n \gets \sum_{i=1}^{k} n_{i}$
			\STATE $J_{i} \gets L_{i}, \mbox{ } \forall i \in \{1, ..., k\}$
			\WHILE {$ Card\big(\bigcup_{i=1}^{k} J_{i}\big) > n$}
			\STATE $J^* \gets \bigcup_{i=1}^{k} J_{i}$
			\STATE $l^{*} \gets \underset{\begin{tiny} \{l\} \in J^*\end{tiny}}{\argmin} \underset{\begin{tiny}\mbox{s.t. \ } card \left( J_{i}\setminus\{l\} \right) \ge n_{i}, \forall i \in \{1, ..., k\}\end{tiny}}{\Gamma_{\delta \alpha \beta}\big(J^*\setminus\{l\}\big)}$
			\FOR{$i \gets 1$ \TO $k$}
			\IF{$l^{*} \in J_{i}$}
			\STATE $J_{i} \gets J_{i}\setminus\{l^{*}\}$
			\ENDIF 
			\ENDFOR
			\ENDWHILE
		\end{algorithmic}
	\end{spacing}
	\vspace{0.1cm}
	\begin{algorithmic}
		\STATE \textbf{Output:} $\{J_{1}, ..., J_k\}$
	\end{algorithmic}
\end{algorithm}

The algorithm proceeds as follows. Starting from an arbitrary partition of a geographical area of interest, expressed as a collection of $k$ locations sets augmented with the minimum number of generating units to be deployed in each of them, $\{(L_1, n_1),..., (L_k, n_k)\}$, the cardinality $n$ of the desired terminal location set is computed. The iterative procedure to trim down the size of the initial locations set to the desired cardinality is then initiated. At each iteration, an individual location is selected and removed from the current locations set, and the cardinality of the latter is compared against the desired cardinality $n$. To be more accurate, the location discarded at each iteration is selected as to minimise the difference (in absolute value) between the criticality index value computed at the previous step and that at the current step. More formally, if $J^*$ stands for the current locations set and $\Gamma^* = \Gamma_{\delta \alpha \beta}\big(J^*\big)$ denotes the value of the criticality indicator associated with the current locations set, the location to be discarded can be obtained as a solution to the following optimisation problem

\begin{linenomath}\begin{align}
l^* = \underset{\{l\} \in J^*}{\argmin} \hspace{10pt} &\big|\Gamma^* - \Gamma_{\delta \alpha \beta}\big(J^*\setminus\{l\}\big)\big| \\
\mbox{s.t. \ } & card \left( J_{i}\setminus\{l\} \right) \ge n_{i}, \ \forall i \in \{1, ..., k\}\nonumber.
\end{align}\end{linenomath}
Exploiting the previously-introduced property of the criticality index to get rid of the absolute value and dropping the constant first term in the objective function yields the more practical form

\begin{linenomath}\begin{align}
l^* = \underset{\{l\} \in J^*}{\argmin} \hspace{10pt} & \Gamma_{\delta \alpha \beta}\big(J^*\setminus\{l\}\big) \\
\mbox{s.t. \ } & card \left( J_{i}\setminus\{l\} \right) \ge n_{i}, \forall i \in \{1, ..., k\}\nonumber.
\end{align}\end{linenomath}
This optimisation problem is solved by enumeration and all individual locations are sorted according to their scores. The location with the minimum score is then selected and discarded, and in case of ties, the first location encountered is arbitrarily dropped. This procedure is repeated until the cardinality of the current locations set matches the desired cardinality $n$, at which point the algorithm terminates.

\section*{References}

\bibliography{Bibliographies}

\begin{thebibliography}{10}
\expandafter\ifx\csname url\endcsname\relax
  \def\url#1{\texttt{#1}}\fi
\expandafter\ifx\csname urlprefix\endcsname\relax\def\urlprefix{URL }\fi
\expandafter\ifx\csname href\endcsname\relax
  \def\href#1#2{#2} \def\path#1{#1}\fi

\bibitem{GiebelPhD}
G.~Giebel, On the benefits of distributed generation of wind energy in
  {E}urope, {PhD} thesis, {U}niversity of {O}ldenburg, {G}ermany (Jul 2001).

\bibitem{LouieChapter}
H.~Louie, J.~M. Sloughter, Probabilistic Modeling and Statistical
  Characteristics of Aggregate Wind Power, Springer Singapore, Singapore, 2014,
  pp. 19--51.

\bibitem{chatzivasileiadis2013global}
S.~Chatzivasileiadis, D.~Ernst, G.~Andersson, The global grid, Renewable Energy
  57 (2013) 372--383.

\bibitem{chatzivasileiadis2017global}
S.~Chatzivasileiadis, D.~Ernst, G.~Andersson, Chapter 12 - {G}lobal {P}ower
  {G}rids for {H}arnessing {W}orld {R}enewable {E}nergy, in: L.~E. Jones (Ed.),
  Renewable Energy Integration (Second Edition), Academic Press, 2017, pp. 161
  -- 174.

\bibitem{hoicka2011solar}
C.~E. Hoicka, I.~H. Rowlands, Solar and wind resource complementarity:
  {A}dvancing options for renewable electricity integration in {O}ntario,
  {C}anada, Renewable Energy 36~(1) (2011) 97 -- 107.

\bibitem{MonfortiComplementarityAssessmentMonteCarlo}
F.~Monforti, T.~Huld, K.~Bodis, L.~Vitali, M.~D'Isidoro, R.~Lacal-Arantegui,
  Assessing complementarity of wind and solar resources for energy production
  in {I}taly. a {M}onte {C}arlo approach, Renewable Energy 63 (2014) 576 --
  586.

\bibitem{WidenCorrelationWindSolarSweden}
J.~Widen, Correlations between large-scale solar and wind power in a future
  scenario for {S}weden, IEEE Transactions on Sustainable Energy 2~(2) (2011)
  177--184.
\newblock \href {http://dx.doi.org/10.1109/TSTE.2010.2101620}
  {\path{doi:10.1109/TSTE.2010.2101620}}.

\bibitem{BettClimatologicalRelationshipsWindSolar}
P.~E. Bett, H.~E. Thornton, The climatological relationships between wind and
  solar energy supply in {B}ritain, Renewable Energy 87 (2016) 96 -- 110.
\newblock \href
  {http://dx.doi.org/https://doi.org/10.1016/j.renene.2015.10.006}
  {\path{doi:https://doi.org/10.1016/j.renene.2015.10.006}}.

\bibitem{JurasComplementarityReliabilitySmallHybrid}
J.~Jurasz, A.~Beluco, F.~A. Canales, The impact of complementarity on power
  supply reliability of small scale hybrid energy systems, Energy 161 (2018)
  737 -- 743.
\newblock \href
  {http://dx.doi.org/https://doi.org/10.1016/j.energy.2018.07.182}
  {\path{doi:https://doi.org/10.1016/j.energy.2018.07.182}}.

\bibitem{CantaoHydroComplementarityBrazil}
M.~P. Cantao, M.~R. Bessa, R.~Bettega, D.~H. Detzel, J.~M. Lima, Evaluation of
  hydro-wind complementarity in the {B}razilian territory by means of
  correlation maps, Renewable Energy 101 (2017) 1215 -- 1225.
\newblock \href
  {http://dx.doi.org/https://doi.org/10.1016/j.renene.2016.10.012}
  {\path{doi:https://doi.org/10.1016/j.renene.2016.10.012}}.

\bibitem{DosAnjosLongTermCrossCorrelationBrazil}
P.~S. dos Anjos, A.~S.~A. da~Silva, B.~Stosic, T.~Stosic, Long-term
  correlations and cross-correlations in wind speed and solar radiation
  temporal series from {F}ernando de {N}oronha {I}sland, {B}razil, Physica A:
  {S}tatistical {M}echanics and its {A}pplications 424 (2015) 90 -- 96.
\newblock \href {http://dx.doi.org/https://doi.org/10.1016/j.physa.2015.01.003}
  {\path{doi:https://doi.org/10.1016/j.physa.2015.01.003}}.

\bibitem{LouieCorrelation}
H.~Louie, Correlation and statistical characteristics of aggregate wind power
  in large transcontinental systems, Wind Energy 17~(6) (2013) 793--810.
\newblock \href {http://dx.doi.org/10.1002/we.1597}
  {\path{doi:10.1002/we.1597}}.

\bibitem{StMartin}
C.~M.~S. Martin, J.~K. Lundquist, M.~A. Handschy, Variability of interconnected
  wind plants: correlation length and its dependence on variability time scale,
  Environmental Research Letters 10~(4).

\bibitem{KatzensteinWindVariability}
W.~Katzenstein, E.~Fertig, J.~Apt, The variability of interconnected wind
  plants, Energy Policy 38~(8) (2010) 4400--4410.

\bibitem{AptPowerSpectrum}
J.~Apt, The spectrum of power from wind turbines, Journal of Power Sources
  169~(2) (2007) 369 -- 374.
\newblock \href
  {http://dx.doi.org/https://doi.org/10.1016/j.jpowsour.2007.02.077}
  {\path{doi:https://doi.org/10.1016/j.jpowsour.2007.02.077}}.

\bibitem{KlimaSolarVariability}
K.~Klima, J.~Apt, Geographic smoothing of solar {PV}: results from {G}ujarat,
  Environmental Research Letters 10~(10) (2015) 104001.

\bibitem{BelucoComplementarityIndexSolarHydro}
A.~Beluco, P.~K. de~Souza, A.~Krenzinger, A dimensionless index evaluating the
  time complementarity between solar and hydraulic energies, Renewable Energy
  33~(10) (2008) 2157 -- 2165.
\newblock \href
  {http://dx.doi.org/https://doi.org/10.1016/j.renene.2008.01.019}
  {\path{doi:https://doi.org/10.1016/j.renene.2008.01.019}}.

\bibitem{LiComplementarityWindSolarOklahoma}
W.~Li, S.~Stadler, R.~Ramakumar, Modeling and assessment of wind and insolation
  resources with a focus on their complementary nature: A case study of
  {O}klahoma, Annals of the Association of American Geographers 101~(4) (2011)
  717--729.
\newblock \href {http://dx.doi.org/10.1080/00045608.2011.567926}
  {\path{doi:10.1080/00045608.2011.567926}}.

\bibitem{gelaro2017modern}
R.~Gelaro, W.~McCarty, M.~J. Su{\'a}rez, R.~Todling, A.~Molod, L.~Takacs, C.~A.
  Randles, A.~Darmenov, M.~G. Bosilovich, R.~Reichle, et~al., The modern-era
  retrospective analysis for research and applications, version 2 ({MERRA}-2),
  Journal of Climate 30~(14) (2017) 5419--5454.

\bibitem{rose2015reanalysiswind}
S.~Rose, J.~Apt, What can reanalysis data tell us about wind power?, Renewable
  Energy 83 (2015) 963--969.

\bibitem{olausson2018eramerra}
J.~Olauson, {ERA}5: {T}he new champion of wind power modelling?, Renewable
  Energy 126 (2018) 322--331.

\bibitem{staffell2016merra}
I.~Staffell, S.~Pfenninger, Using bias-corrected reanalysis to simulate current
  and future wind power output, Energy 114 (2016) 1224--1239.

\bibitem{pfenninger2016merra}
I.~Staffell, S.~Pfenninger, Long-term patterns of {E}uropean {PV} output using
  30 years of validated hourly reanalysis and satellite data, Energy 114 (2016)
  1251--1265.

\bibitem{ecmwf2018era5}
ECMWF, Copernicus knowledge base - {ERA}5 data documentation,
  https://confluence.ecmwf.int//display/CKB/ (2018).

\bibitem{fettweis2017reconstructions}
X.~Fettweis, J.~E. Box, C.~Agosta, C.~Amory, C.~Kittel, C.~Lang, D.~van As,
  H.~Machguth, H.~Gall{\'e}e, Reconstructions of the 1900--2015 {G}reenland ice
  sheet surface mass balance using the regional climate {MAR} model, The
  Cryosphere 11~(2) (2017) 1015.

\bibitem{radu2018greenland}
D.~Radu, M.~Berger, R.~Fonteneau, X.~Fettweis, S.~Hardy, D.~Ernst,
  Complementarity assessment of {S}outh {G}reenland katabatic flows and {W}est
  {E}urope wind regimes, Submitted.

\bibitem{aerodyn_scd_8MW}
{A}erodyn~{E}ngineering {GMBH}, {SCD} 8.0 {MW} – {T}echnical {D}ata, Tech.
  rep., {A}erodyn {E}ngineering {GMBH} (2018).

\bibitem{obermann2016windsfrance}
A.~Obermann, S.~Bastin, S.~Belamari, D.~Conte, M.~A. Gaertner, L.~Li,
  B.~Ahrens, Mistral and {T}ramontane wind speed and wind direction patterns in
  regional climate simulations, Climate Dynamics 51 (2016) 1059--1076.

\bibitem{rte2017bilan}
{RTE}, Bilan électrique et perspectives - {B}retagne/{H}auts de
  {F}rance/{O}ccitanie/{N}ormandie (2017).

\end{thebibliography}

\end{document}